\documentclass{article}

\usepackage{mathtools,amssymb,amsthm,amsfonts}
\usepackage{graphicx,color}
\usepackage{booktabs}
\usepackage{microtype}
\usepackage{xifthen}
\usepackage{dsfont}
\usepackage{bm}
\usepackage[boxed]{algorithm2e}
\usepackage[numbers,sort&compress]{natbib}
\usepackage{hyperref}
\usepackage[capitalize]{cleveref}

\theoremstyle{plain}
\newtheorem{theorem}{Theorem}[section]

\newtheorem{assumption}[theorem]{Assumption}

\newtheorem{corollary}[theorem]{Corollary}

\theoremstyle{definition}

\newtheorem{definition}[theorem]{Definition}

\theoremstyle{remark}

\newtheorem{remark}[theorem]{Remark}

\Crefname{assumption}{Assumption}{Assumptions}
\crefname{remark}{Remark}{Remarks}
\Crefname{remark}{Remark}{Remarks}
\crefname{corollary}{Corollary}{Corollaries}
\Crefname{corollary}{Corollary}{Corollaries}

\title{Reweighted information inequalities}
\author{Jonathan~Niles-Weed\footnote{Departments of Mathematics and Data Science, Courant Institute School, New York University. \texttt{jnw@cims.nyu.edu}}}
\date{}

\newcommand{\MR}[1]{}
\newcommand{\doi}[1]{\href{https://doi.org/#1}{\texttt{doi:#1}}}

\newcommand{\cC}{\mathcal{C}}

\newcommand{\cE}{\mathcal{E}}

\newcommand{\cL}{\mathcal{L}}

\newcommand{\cN}{\mathcal{N}}

\newcommand{\cP}{\mathcal{P}}

\newcommand{\cX}{\mathcal{X}}

\newcommand{\RR}{\mathbb{R}}

\newcommand{\1}{\mathds{1}}

\newcommand*{\kl}[3][]{\ifthenelse{\isempty{#1}}{\operatorname{D}(#2\,\|\,#3)}{\operatorname{D}(#2\,\|\,#3\mid#1)}}

\DeclarePairedDelimiter{\triplenorm}{\vert\kern-0.25ex\vert\kern-0.25ex\vert}{\vert\kern-0.25ex\vert\kern-0.25ex\vert}

\newcommand*{\E}{\mathbb E}

\DeclareMathOperator{\var}{Var}

\newcommand*{\ep}{\varepsilon}

\newcommand*{\rd}{\mathrm{d}}
\newcommand*{\dd}{\, \rd}

\newcommand*{\todo}[1][]{\ifthenelse{\equal{#1}{}}{\textcolor{red}{[\textbf{TODO}]}}{\textcolor{red}{[\textbf{TODO:} #1]}}}

 \newcommand{\FI}[2]{\operatorname{FI}(#1\, \| \, #2)}

\DeclareMathOperator{\Dom}{Dom}
\newcommand{\tii}{T_cI}

\newcommand{\radon}[2]{\frac{\rd #1}{\rd #2}}

\begin{document}
\maketitle

\begin{abstract}
We establish a variant of the log-Sobolev and transport-information inequalities for mixture distributions. If a probability measure $\pi$ can be decomposed into components that individually satisfy such inequalities, then any measure $\mu$ close to $\pi$ in relative Fisher information is close in relative entropy or transport distance to a reweighted version of $\pi$ with the same mixture components but possibly different weights. This provides a user-friendly interpretation of Fisher information bounds for non-log-concave measures and explains phenomena observed in the analysis of Langevin Monte Carlo for multimodal distributions.
\end{abstract}

\section{Introduction}
Log-Sobolev inequalities and their close cousins, transport-information inequalities, give important information about the convergence properties of Markov processes.
Consider the overdamped Langevin dynamics
\begin{equation}\label{eq:langevin}
	dX_t = - \nabla V(X_t) dt + \sqrt{2}d B_t
\end{equation}
for some smooth $V: \RR^d \to \RR$ satisfying $\int e^{-V} \dd x < \infty$.
Define the relative entropy
\begin{equation}\label{eq:kl}
	\kl{\mu}{\pi} := \begin{cases}
		\int \log \frac{\rd \mu}{\rd \pi} \dd \mu & \text{if $\mu \ll \pi$,} \\
		\infty & \text{otherwise.}
	\end{cases}
\end{equation}
If the stationary distribution $\pi \propto e^{-V}$ satisfies the log-Sobolev inequality\footnote{The choice of constant is to ensure consistency with the general definition of Fisher information introduced later.}~\citep{gross}
\begin{equation}\label{eq:basic_lsi}
	\kl{\mu}{\pi} \leq c \FI{\mu}{\pi} := \frac{c}{4} \int \frac{\|\nabla \rd \mu/\rd \pi\|^2}{\rd \mu/\rd \pi} \dd \pi \quad \quad \forall \mu \in \cP(\RR^d),\, \mu \ll \pi\,,
\end{equation}
then the law $\mu_t$ of $X_t$ converges to $\pi$ exponentially fast.
Indeed, De Bruijn's identity~\citep{stam} implies that the relative Fisher information appearing on the right side of~\eqref{eq:basic_lsi} is precisely the rate of dissipation of relative entropy under Langevin dynamics:
\begin{equation}\label{eq:dissip}
	\frac{d}{dt} \kl{\mu_t}{\pi} = - 4 \FI{\mu_t}{\pi}\,.
\end{equation}
Exponential decay then follows from Gronwall's inequality (see, e.g., \citep[Appendix B]{evans-pde}).
Moreover, the log-Sobolev inequality in fact implies the transport-information inequality~\citep{otto-villani, guillin-leonard-wu-yao}
\begin{equation}\label{eq:basic_ti}
	W_2^2(\mu, \pi) \leq c \kl{\mu}{\pi} \leq c^2 \FI{\mu}{\pi}\,,
\end{equation}
so that exponential convergence holds in Wasserstein distance as well.

The inequalities~\eqref{eq:basic_lsi} and~\eqref{eq:basic_ti} are fundamental to the analysis of Langevin dynamics, but their applicability is limited: the most common examples of measures satisfying such inequalities are strongly log-concave probability distributions~\citep{bakry-emery}.
Conversely, such estimates do not hold at all (or only with constants scaling exponentially in the dimension) for measures that are ``far'' from log-concave (e.g., when $V$ has multiple local minima).
Langevin dynamics for such measures are known to mix exponentially slowly~\citep{bovier-eckhoff-gayrard-klein}.

Nevertheless, integrating the dissipation relation~\eqref{eq:dissip} shows that, no matter the geometry of $\pi$, average Fisher information decreases along the Langevin trajectory:
\begin{equation}
	\frac{1}{t} \int_0^t \FI{\mu_s}{\pi} \dd s \leq \frac{\kl{\mu_0}{\pi}}{4t}\,.
\end{equation}
\citet{BalCheErd22} exploited this to show that, for any $\pi \propto e^{-V}$ with sufficiently smooth $V$, Langevin Monte Carlo produces a measure $\mu$ satisfying $\FI{\mu}{\pi} \leq \ep$ in time polynomial in $d$ and $1/\ep$, without any convexity assumptions on $V$.
This raises the question: what does $\FI{\mu}{\pi} \leq \ep$ imply about the geometry of $\mu$ vis-\`a-vis $\pi$ in the absence of a log-Sobolev or transport-information inequality?
The goal of this note is to provide a user-friendly answer under a decomposability assumption on $\pi$.

To give a taste of the main results, it is useful to review an example from \citet{BalCheErd22} in detail.
They consider a mixture of one-dimensional Gaussians with well-separated modes: $\pi = \frac 12 \cN(-m, 1) + \frac 12 \cN(m, 1)$ for $m \gg 0$.
This is a paradigmatic example of a measure that fails to satisfy a log-Sobolev or transport-information inequality with good constant: as $m \to \infty$, the best constant $c$ in~\eqref{eq:basic_lsi} scales as $e^{(1+o(1))m^2/2}$~\citep{Chafai-Malrieu}.
By a direct calculation, \citet{BalCheErd22} exhibit a measure $\mu = \frac{3}{4}\cN(-m, 1) + \frac 1 4 \cN(m, 1)$ whose relative entropy to $\pi$ is large but whose relative Fisher information is exponentially small.
Though far from $\pi$, the measure $\mu$ is a \textit{reweighted} version of $\pi$ that matches its local structure; Langevin dynamics initialized in a neighborhood of $\mu$ remain there for exponentially long before reaching $\pi$.

The main theorems below explain this phenomenon in much greater generality.
For instance, let $\pi = \sum_{i=1}^m w_i \cN(m_i, I)$ be an arbitrary mixture of isotropic Gaussians in any dimension.
Such a measure will typically only enjoy a transport-information or log-Sobolev inequality with constant exponentially large in the parameters of the problem~\citep{chen-chewi-niles-weed}, so no inequality of the form~\eqref{eq:basic_ti} with good constant can hold.
However, our main theorem implies that such a measure \textit{does} enjoy the inequalities
\begin{equation}\label{eq:basic_reweighted}
	\inf_{\tilde \pi \in \cC} \kl{\mu}{\tilde \pi} \leq 2 \FI{\mu}{\pi}\,, \quad \quad \quad \inf_{\tilde \pi \in \cC} W_2^2(\mu, \tilde \pi) \leq 4 \FI{\mu}{\pi}\,,
\end{equation}
for any probability measure $\mu$, where $\cC := \{\sum_{i=1}^m \tilde w_i \cN(m_i, I) : \tilde w \in \Delta_m\}$ is the set of \textit{reweighted} versions of $\pi$.
In other words, any measure close to $\pi$ in relative Fisher information is automatically close to a modification of $\pi$ with the same mixture components but possibly different mixture weights; the proofs identify a specific $\tilde \pi \in \cC$ achieving the bound, whose weights essentially agree with the mass that $\mu$ assigns to the neighborhood of each component; see~\eqref{eq:lambda}.
This validates and generalizes the intuition of \citet{BalCheErd22}.

The main theorems are not limited to Gaussian mixtures or Langevin diffusion: they apply to general mixtures and Markov semigroups.
If $\pi$ decomposes as a mixture of measures satisfying log-Sobolev or transport-information inequalities, then $\pi$ enjoys reweighted versions of these inequalities analogous to~\eqref{eq:basic_reweighted}.
Though the proofs are elementary, precise statements do not seem to appear in the literature.
These inequalities are closely connected to higher-order spectral gaps for generators of Markov processes developed by~\citet{KoeLeeVuo24} for non-log-concave sampling; our main results may be seen as non-linear analogues of their spectral bounds for mixtures.\footnote{I am grateful to Sinho Chewi for this observation.}
\section{Preliminaries and prior work}\label{sec:prelim}
To state our results, we establish some notation.
We work throughout with Borel probability measures on a complete separable metric space $(\cX, d)$.
Denote by $\cL$ and $\{\cL_i\}_{i \in [m]}$ the infinitesimal generators of symmetric, ergodic Markov semigroups with invariant distributions $\pi$ and $\{\pi_i\}_{i \in [m]}$, respectively.
We define the bilinear \textit{Dirichlet form} $\cE$ associated with $\cL$ to be the $L^2(\pi)$ closure of
\begin{equation*}
	\cE(f, f) := - \int f \cL f \dd \pi\,, \quad f \in \Dom(\cL)\,,
\end{equation*}
and define the Dirichlet forms $\{\cE_i\}_{i \in [m]}$ associated with $\{\cL_i\}_{i \in [m]}$ analogously.
In the important special case where $\cL = \Delta - \langle \nabla V, \nabla \rangle$ is the generator of the Langevin diffusion~\eqref{eq:langevin}, this definition reads
\begin{equation*}
	\cE(f, f) = \int |\nabla f|^2 \dd \pi\,, \quad f \in H^1(\pi)\,.
\end{equation*}
We refer the reader to~\citet{bakry-gentil-ledoux} for background on Markov semigroups and their associated Dirichlet forms.

Our main assumption is that $\pi$ can be decomposed as a mixture of the measures $\{\pi_i\}_{i \in [m]}$, and that the corresponding Dirichlet forms are compatible.
To that end, we borrow the following assumption from~\citet{KoeLeeVuo24}.
\begin{assumption}\label[assumption]{assume:main}
	The measures $\pi_1, \dots, \pi_m$ and $\pi$ satisfy
	\begin{align*}
		\pi  = \sum_{i=1}^m w_i \pi_i \quad \text{and} \quad
\sum_{i=1}^m w_i \cE_i(f, f) \leq \cE(f, f)  \quad \forall f \in \Dom(\cE)
	\end{align*}
	for some strictly positive weights $w_1, \dots, w_m$.
\end{assumption}
The second requirement of~\cref{assume:main} is satisfied when $\cL$ and $\{\cL_i\}_{i \in [m]}$ are the generators of overdamped Langevin or Glauber dynamics~\citep[Lemma 9]{KoeLeeVuo24}.
In particular, our results apply equally well to sampling from continuous measures on $\RR^d$ and from Ising models or other discrete models on the binary hypercube $\{\pm 1\}^n$.
Some of our results (specifically, those involving the modified log-Sobolev inequality) require the stronger condition $\sum_{i=1}^m w_i \cE_i(f, \log f) \leq \cE(f, \log f)$; for overdamped Langevin dynamics, this holds as an equality since the carr\'e du champ $\Gamma(f, g) = \langle \nabla f, \nabla g \rangle$ does not depend on the invariant measure.
For Glauber dynamics, this inequality also holds by the same argument used in the proof of~\cite[Fact A.3]{huang-mohanty-rajaraman-wu}, since $(\log x- \log y)(x - y)$ is always nonnegative.

The key object is the \textit{Fisher information}.
\begin{definition}
	Given the Dirichlet form $\cE$ with domain $\Dom(\cE) \subseteq L^2(\pi)$, the relative Fisher information of $\mu$ with respect to $\pi$ is
	\begin{equation}
		\FI{\mu}{\pi} = \begin{cases}
			\cE(\sqrt f, \sqrt f) & \text{if $\mu = f \pi$ and $\sqrt f \in \Dom(\cE)$,} \\
			+\infty & \text{otherwise.}
		\end{cases}
	\end{equation}
	The relative Fisher information with respect to $\pi_i$ is defined analogously.
\end{definition}

Transport-information inequalities relate the Fisher information to a transportation cost.
We always consider costs $c : \cX^2 \to [0, +\infty]$ which are lower-semicontinuous.
\begin{definition}
	The transportation cost $T_c(\mu, \pi)$ corresponding to the cost $c$ is
	\begin{equation}
		T_c(\mu, \pi) = \inf_{\gamma \in \Gamma(\mu, \pi)} \int c(x, y) \dd \gamma(x, y)\,,
	\end{equation}
	where $\Gamma(\mu, \pi)$ denotes the set of all couplings between $\mu$ and $\pi$.

	Fix an increasing, convex, lower-semicontinuous function $\alpha: [0, +\infty) \to [0, +\infty]$ satisfying $\alpha(0) = 0$.
	The measure $\pi_i$ satisfies the transportation-information inequality $\tii$ if
	\begin{equation}\label{eq:ti}\tag{$\tii$}
		\alpha(T_c(\mu, \pi_i)) \leq \FI{\mu}{\pi_i} \quad \forall \mu \in \cP(\cX)\,.
	\end{equation}
\end{definition}

Log-Sobolev inequalities relate Fisher information to relative entropy.
\begin{definition}
	The measure $\pi_i$ satisfies a log-Sobolev inequality if there exists a finite $C \geq 0$ such that
	\begin{equation}\label{eq:lsi}\tag{LSI}
		\kl{\mu}{\pi_i} \leq C \FI{\mu}{\pi_i} \quad \forall \mu \in \cP(\cX)\,.
	\end{equation}
	Similarly, $\pi_i$ satisfies a modified log-Sobolev inequality~\citep{bobkov-tetali} if there exists a finite $C' \geq 0$ such that
	\begin{equation}\label{eq:mlsi}\tag{MLSI}
	\kl{\mu}{\pi_i} \leq C' \cE\left(\frac{\rd \mu}{\rd \pi_i}, \log \frac{\rd \mu}{\rd \pi_i}\right)\,.
	\end{equation}
\end{definition}
The definition of the modified log-Sobolev inequality is motivated by the fact that $\cE\left(\frac{\rd \mu_t}{\rd \pi}, \log \frac{\rd \mu_t}{\rd \pi}\right)$ is the time derivative of $\kl{\mu_t}{\pi}$ under the evolution $\frac{d \mu_t}{dt} = \cL^* \mu_t$ given by the Markov semigroup.
It is easy to see that for Langevin diffusion, $\FI{\mu}{\pi} = \tfrac 1 4 \cE\left(\frac{\rd \mu}{\rd \pi}, \log \frac{\rd \mu}{\rd \pi}\right)$, which justifies~\eqref{eq:dissip} and shows that \eqref{eq:lsi} and \eqref{eq:mlsi} are equivalent.
For general Markov semigroups, however, \eqref{eq:lsi} implies \eqref{eq:mlsi} with $C' = C/4$ but not conversely.

\subsection{Prior work}
The literature on functional inequalities and rates of convergence for Markov semigroups is vast.
Transport-information inequalities were first systematically studied by~\citet{guillin-leonard-wu-yao}, though their natural connection to Langevin diffusion and the geometry of Wasserstein space was highlighted by~\citet{otto-villani}.
Log-Sobolev inequalities were initially developed by~\citet{gross} to study the hypercontractivity phenomenon and have been a central tool in the convergence analysis of Markov processes~\citep{bakry-gentil-ledoux}.
The approach of decomposing a measure into simpler components to bound functional inequality constants goes back to~\citet{jerrum-son-tetali-vigoda}.
\citet{schlichting-mixtures} studied such inequalities for mixtures directly, obtaining an entropy decomposition into within-component and coarse-grained terms closely related to the decomposition used in \cref{thm:rlsi}.

From the viewpoint of applications, a wealth of recent work has focused on algorithmic approaches to sampling from non-log-concave measures, including via generative modeling.
This note fits into a recent line of work that uses the assumption that the target measure decomposes into a mixture of well-behaved components to obtain positive guarantees~(see, e.g., \citep{ge-lee-risteski}).
Closest to our setting is the work of~\citet{KoeLeeVuo24}: they show that a decomposition into components satisfying a Poincar\'e inequality automatically yields a ``higher-order spectral gap'' for the generator $\cL$.
This result implies mixing guarantees so long as the initialization approximately assigns the correct mass to each component; see \citep{huang-mohanty-rajaraman-wu} for a similar guarantee.
As an important application of their results, they show that ``low-complexity Gibbs measures'' possess such decompositions, building on ideas of~\citet{Koehler-lee-risteski}.
A benefit of our approach is that it gives guarantees involving relative entropy (via log-Sobolev inequalities) rather than $\chi^2$-divergence (which follow from Poincar\'e inequalities).
In some sampling applications, the former guarantee can be exponentially more informative than the latter (see \citep{chewi-log-concave-sampling}).

As we explore in \cref{sec:meta}, our results can be used to obtain convergence guarantees for Langevin diffusions and other dynamics that show that the iterates quickly match the local geometry of the target measure.
A similar phenomenon was investigated by~\citet{cheng-wang-zhang-zhu}, who showed that Langevin Monte Carlo enjoys a ``conditional mixing'' guarantee on any subset on which $\pi$ satisfies a log-Sobolev or Poincar\'e inequality.

\subsection{Notation}
We write $\cP(\cX)$ for the space of Borel probability measures on $\cX$.
The probability simplex in $\RR^m$ is denoted by $\Delta_m$, which we also identify with the set $\cP([m])$ of probability measures on $[m] := \{1, \dots, m\}$.
The relative entropy $\kl{\cdot}{\cdot}$ is defined for any pair of Borel probability measures on a Polish space by~\eqref{eq:kl}.
Throughout, we write $\cC := \{\sum_{i=1}^m \lambda_i \pi_i: \lambda \in \Delta_m\}$ for the set of ``reweighted mixtures;'' equivalently, this is the convex hull of $\pi_1, \dots, \pi_m$.
We adopt the shorthand notation
\begin{align*}
	T_c(\mu, \cC) := \inf_{\tilde \pi \in \cC} T_c(\mu, \tilde \pi)\,. \quad \quad
	\kl{\mu}{\cC} := \inf_{\tilde \pi \in \cC} \kl{\mu}{\tilde \pi}\,.
\end{align*}
for transport cost and relative entropy between $\mu$ and the closest element of $\cC$.
\section{Main results}
Our first main result shows that if the component measures $\pi_1, \dots, \pi_m$ satisfy a transport-information inequality, then the mixture $\pi$ satisfies a \textit{reweighted} transport-information inequality.
\begin{theorem}\label{thm:main}
	Let $\pi_1, \dots, \pi_m$ and $\pi$ satisfy \Cref{assume:main}.
	If $\pi_1, \dots, \pi_m$ satisfy~\eqref{eq:ti}, then
	\begin{equation}\label{eq:rti}
		\alpha(T_c(\mu, \cC)) = \alpha\left(\inf_{\tilde \pi \in \cC} T_c(\mu, \tilde \pi)\right) \leq \FI{\mu}{\pi} \quad \forall \mu \in \cP(\cX)\,.
	\end{equation}
\end{theorem}

As we observed above, \cref{thm:main} implies that any probability measure $\mu$ close to $\pi$ in relative Fisher information is close in transport distance to a version of $\pi$ with the same mixture components but possibly different mixture weights.

The proof identifies a particular element of $\cC$ for which the inequality holds: we can take $\tilde \pi = \sum_{i=1}^m \lambda^*_i \pi_i$, where
\begin{equation}\label{eq:lambda}
	\lambda^*_i := w_i \E_{\pi_i} \frac{\rd \mu}{\rd \pi}\,.
\end{equation}

As our proof reveals, the same argument holds when $\alpha \circ T_c$ is replaced by any jointly convex function.
We will exploit this fact later in the development of reweighted log-Sobolev inequalities.
\begin{proof}
	First, note that $\alpha(T_c(\cdot, \cdot))$ is jointly convex.
	Indeed, $T_c$ is convex~\citep[Theorem 4.8]{Vil08} and nonnegative, and the composition of a convex function and an increasing convex function is convex.
	
	Consider an arbitrary convex decomposition $\mu = \sum_{i=1}^m \lambda_i \mu_i$ for $\lambda \in \Delta_m$ and $\mu_1, \dots, \mu_m \in \cP(\cX)$.
	Then convexity of $\alpha \circ T_c$ implies
	\begin{align}
		\alpha(T_c(\mu, \cC))  =  \alpha\left(T_c\left(\sum_{i=1}^m \lambda_i \mu_i, \cC\right)\right)
		& \leq \alpha\left(T_c\left(\sum_{i=1}^m \lambda_i \mu_i, \sum_{i=1}^m \lambda_i \pi_i\right)\right) \nonumber \\
		& \leq \sum_{i=1}^m \lambda_i \alpha(T_c(\mu_i, \pi_i)) \label{eq:reweighted_ub}\,.
	\end{align}
	
	We now produce a particular decomposition of $\mu$ for which the resulting expression can be bounded by $\FI{\mu}{\pi}$.
	Under \Cref{assume:main}, $\pi_i \ll \pi$, and we may assume that $\mu \ll \pi$ since otherwise the claim is vacuous.
Then $\frac{\rd \mu}{\rd \pi} \in L^1(\pi) \subseteq L^1(\pi_i)$ (since $\pi_i \leq \pi/w_i$), and we can define the finite measure $\bar \mu_i(\rd x) := w_i \frac{\rd \mu}{\rd \pi}(x) \pi_i(\rd x)$.
The total mass of $\bar \mu_i$ is $\lambda_i^*$, as defined in~\eqref{eq:lambda}.
Renormalizing, we obtain the probability measure
	\begin{equation}\label{eq:mu_def}
	\mu_i = \begin{cases}
		\frac{\bar \mu_i}{\lambda^*_i} & \text{if $\lambda^*_i > 0$,} \\
		\pi_i & \text{otherwise.}
	\end{cases}
	\end{equation}
	Since $\sum_{i=1}^m w_i \pi_i = \pi$, this furnishes a valid decomposition: $\lambda^* \in \Delta_m$ and $\sum_{i=1}^m \lambda^*_i \mu_i = \mu$.
	By construction, $\mu_i \ll \pi_i$, and
	\begin{equation}\label{eq:rn_key}
		\lambda^*_i \frac{\rd \mu_i}{\rd \pi_i} = w_i \frac{\rd \mu}{\rd \pi}
	\end{equation}
	as elements of $L^1(\pi_i)$.
	In particular, under \Cref{assume:main}, if $\sqrt{\frac{\rd \mu}{\rd \pi} } \in \Dom(\cE)$, then $\sqrt{\frac{\rd \mu_i}{\rd \pi_i}} \in \Dom(\cE_i)$.

	Applying \eqref{eq:ti} and~\eqref{eq:rn_key} and using \Cref{assume:main} yields
	\begin{align*}
		\sum_{i=1}^m \lambda^*_i \alpha(T_c(\mu_i, \pi_i)) & \leq \sum_{i=1}^m \lambda^*_i \cE_i(\sqrt{\frac{\rd \mu_i}{\rd \pi_i}},\sqrt{\frac{\rd \mu_i}{\rd \pi_i}}) \\
		& = \sum_{i=1}^m w_i \cE_i(\sqrt{\frac{\rd \mu}{\rd \pi}},\sqrt{\frac{\rd \mu}{\rd \pi}}) \\
		& \leq \cE(\sqrt{\frac{\rd \mu}{\rd \pi}},\sqrt{\frac{\rd \mu}{\rd \pi}}) = \FI{\mu}{\pi}\,.
	\end{align*} 
\end{proof}

\begin{remark}\label[remark]{prob_interp}
	The construction used in the proof of \cref{thm:main} has a natural probabilistic interpretation in terms of ``labeled'' mixtures.
	Consider the lifted probability measure $\bm{\pi} \in \cP([m] \times \cX)$ defined by
	\begin{equation}
		\bm{\pi} = \sum_{i=1}^m w_i \delta_i \otimes \pi_i\,.
	\end{equation}
	The measure $\bm{\pi}$ represents the joint distribution of a sample from $\pi$ and the component from which it comes.
	Note that if $(I, X) \sim \bm{\pi}$, then marginally $X \sim \pi$ and $I \sim w$.
	Define $\bm{\mu} \in \cP([m] \times \cX)$ by $\bm{\mu}(\rd i, \rd x) = \mu(\rd x) \bm{\pi}(\rd i | x)$.
	Then the vector $\lambda^*$ defined in the proof of \cref{thm:main} is the marginal distribution of $I$ under $\bm{\mu}$, and $\mu_i$ is the conditional distribution under $\bm{\mu}$ of $X$ given $I = i$.
	In the language of the expectation-maximization algorithm, the vector $\lambda^*$ comprises the posterior weights or ``responsibilities'' obtained when fitting the mixture model $\pi$~\citep{dempster-laird-rubin}.
\end{remark}
The same technique applies to log-Sobolev inequalities.

\begin{theorem}\label{thm:rlsi}
	Let $\pi_1, \dots, \pi_m$ and $\pi$ satisfy \Cref{assume:main}.
	If $\pi_1, \dots, \pi_m$ satisfy \eqref{eq:lsi}, then
	\begin{equation}
		\kl{\mu}{\cC} = \inf_{\tilde \pi \in \cC} \kl{\mu}{\tilde \pi} \leq \kl{\mu}{\pi} -
		\kl{\lambda^*}{w}  \leq
		C \FI{\mu}{\pi} \quad \forall \mu \in \cP(\cX)\,,
	\end{equation}
	where $\lambda^* \in \Delta_m$ is defined by~\eqref{eq:lambda}.\footnote{To be explicit, the quantity $\kl{\lambda^*}{w}$ is a relative entropy between the vectors $\lambda^*$ and $w$ viewed as elements of $\cP([m])$.}
	Analogously, if $\pi_1, \dots, \pi_m$ satisfy \eqref{eq:mlsi} and
	\begin{equation}\label{eq:dirichlet_entropy}
		\sum_{i=1}^m w_i \cE_i(f, \log f) \leq \cE(f, \log f) \quad \forall f > 0,\, f \in \Dom(\cE)\,,
	\end{equation}
	then
	\begin{equation}
		\kl{\mu}{\cC} = \inf_{\tilde \pi \in \cC} \kl{\mu}{\tilde \pi} \leq\kl{\mu}{\pi} -
		\kl{\lambda^*}{w} \leq C' \cE\left(\frac{\rd \mu}{\rd \pi}, \log \frac{\rd \mu}{\rd \pi}\right)\,.
	\end{equation}
\end{theorem}

As in \cref{thm:main}, the bounds hold for the particular element $\tilde \pi = \sum_{i=1}^m \lambda_i^* \pi_i$.
The quantity $\kl{\mu}{\pi} -
\kl{\lambda^*}{w}$ has a natural probabilistic interpretation in the setting of \Cref{prob_interp}: it is the conditional relative entropy $\kl[\bm{\mu}_I]{\bm{\mu}_{X \mid I}}{\bm{\pi}_{X \mid I}}$.

\begin{proof}
	The relative entropy $\kl{\cdot}{\cdot}$ is jointly convex in its arguments~\citep{cover-thomas}, hence as in~\eqref{eq:reweighted_ub}, we have for the decomposition defined by~\eqref{eq:lambda} and~\eqref{eq:mu_def}
	\begin{equation}\label{eq:kl_convexity}
		\kl{\mu}{\cC}  \leq \sum_{i=1}^m \lambda^*_i \kl{\mu_i}{\pi_i}\,.
	\end{equation}
	It suffices to show that the right side of~\eqref{eq:kl_convexity} agrees with $\kl{\mu}{\pi} -
	\kl{\lambda^*}{w}$, since the rest of the proof of \cref{thm:main} applies verbatim to show that $C \FI{\mu}{\pi}$ is an upper bound.
	This follows directly from~\eqref{eq:rn_key}:
	\begin{align*}
		\kl{\mu}{\pi} & = \int \log \frac{\rd \mu}{\rd \pi}(x) \,\mu(\rd x) \\
		& = \sum_{i=1}^m \lambda_i^* \int \log \frac{\rd \mu}{\rd \pi}(x) \, \mu_i(\rd x) \\
		& = \sum_{i=1}^m \lambda_i^* \int \log \left(\frac{\lambda^*_i}{w_i} \frac{\rd \mu_i}{\rd \pi_i}(x)\right) \, \mu_i(\rd x) \\
		& = \kl{\lambda^*}{w} + \sum_{i=1}^m \lambda_i^* \kl{\mu_i}{\pi_i}
	\end{align*}

	For the second claim, we employ the same proof along with the fact that $f \mapsto \cE(f, \log f)$ is $1$-homogeneous, so by \eqref{eq:rn_key},\begin{equation}
		\sum_{i=1}^m \lambda^*_i \cE_i \left(\frac{\rd \mu_i}{\rd \pi_i}, \log \frac{\rd \mu_i}{\rd \pi_i}\right) = \sum_{i=1}^m w_i \cE_i \left(\frac{\rd \mu}{\rd \pi}, \log \frac{\rd \mu}{\rd \pi}\right)\,.
	\end{equation}
	The conclusion then follows from~\eqref{eq:dirichlet_entropy}.
\end{proof}

\begin{remark}\label[remark]{rmk:talagrand}
	If each $\pi_i$ satisfies a Talagrand inequality $T_c(\mu, \pi_i) \leq C \kl{\mu}{\pi_i}$, then the same argument yields a reweighted Talagrand inequality:
	$T_c(\mu, \cC) \leq C (\kl{\mu}{\pi} - \kl{\lambda^*}{w}) \leq C \kl{\mu}{\pi}$.
\end{remark}
It is well known that the log-Sobolev inequality linearizes to the Poincar\'e inequality (see, e.g., \citep[Section~5.1]{bakry-gentil-ledoux}):
\begin{equation}\label{eq:real_poin}
	\var_\pi(g) \leq C \cE(g, g) \quad \forall g \in \Dom(\cE)\,.
\end{equation}
Similarly, linearizing the proof of \cref{thm:rlsi} establishes the following reweighted Poincar\'e inequality.
\begin{corollary}\label[corollary]{cor:poin}
	Let $\pi_1, \dots, \pi_m$ and $\pi$ satisfy \Cref{assume:main}.
	 If $\pi_1, \dots, \pi_m$ satisfy~\eqref{eq:real_poin}, then
	\begin{equation}\label{eq:rpoin}
		\inf_{\tilde g \in \operatorname{span}\left(\left\{\frac{\rd \pi_i}{\rd \pi}\right\}_{i=1}^m\right)}\var_\pi(g - \tilde g) \leq C \cE(g, g) \quad \forall g \in \Dom(\cE)
	\end{equation}
\end{corollary}
We omit the proof, since this result was already obtained by a different argument in~\citep[Lemma 9]{KoeLeeVuo24}.
The infimum on the left side of~\eqref{eq:rpoin} is nothing but the squared $L^2(\pi)$ norm of the linear projection of $g$ onto the orthogonal complement of $ \operatorname{span}\left(\left\{\frac{\rd \pi_i}{\rd \pi}\right\}_{i=1}^m\right)$, justifying the observation made above that
\Cref{thm:rlsi} can be viewed as a nonlinear version of \Cref{cor:poin}.

We record another version of the reweighted Poincar\'e inequality, which controls the total variation metric.
\begin{corollary}
		Let $\pi_1, \dots, \pi_m$ and $\pi$ satisfy \Cref{assume:main}.
		If $\pi_1, \dots, \pi_m$ satisfy~\eqref{eq:real_poin}, then
		\begin{equation*}
			\inf_{\tilde \pi \in \cC} \|\mu - \tilde \pi\|_{\mathrm{TV}}^2 \leq 4 C \FI{\mu}{\pi}\,.
		\end{equation*}
\end{corollary}
\begin{proof}
	\citet[Theorem 3.1]{guillin-leonard-wu-yao} show that~\eqref{eq:real_poin} implies that $\pi_i$ satisfies a transport-information inequality $\|\mu - \pi_i\|_{\mathrm{TV}}^2 \leq 4 C \FI{\mu}{\pi_i}$.
	The result then follows from \cref{thm:main}.
\end{proof}

\section{Metastability and exponentially fast convergence to a reweighted mixture}\label{sec:meta}
In this section, we consider the special case of mixtures whose components satisfy~\eqref{eq:mlsi}.
Denote by $\mu_t$ the evolution of a measure $\mu$ under the Markov semigroup with generator $\cL$, initialized from $\mu_0 = \mu$.

If $\pi$ itself does not satisfy~\eqref{eq:mlsi}, the convergence to equilibrium can be slow.
For example, in the case of Langevin dynamics with non-convex $V$, a particle can remain trapped in the neighborhood of a local minimum of $V$ for exponentially long times.
The concept of metastability, whose origins go back to~\citep{Eyring, Kramers, freidlin-wentzell} and whose modern formulation is due to~\citet{bovier-eckhoff-gayrard-klein, bovier-gayrard-klein}, makes this phenomenon precise and connects it with the setting we study here: a metastable system is characterized by a decomposition of state space into nearly invariant regions, transitions between which happen at large ``metastable'' time scales.

Such a decomposition falls neatly into the framework we have developed for reweighted log-Sobolev inequalities.
\citet{menz-schlichting} show that if $\pi \propto e^{-H/\ep}$ for some potential $H$ satisfying suitable smoothness and growth conditions, then there exists a finite partition $\{\Omega_i\}_{i=1}^m$ of $\RR^d$, which approximately corresponds to basins of attraction of local minima of $H$, such that the conditional measures $\pi_i \propto \pi \1_{\Omega_i}$ each satisfy a log-Sobolev inequality with constant $C = O(1)$ as $\ep \to 0$.
Moreover, the waiting time between transitions between different elements in the partition is exponentially long; equivalently, $\mu_0(\Omega_i) \approx \mu_t(\Omega_i)$ up to exponentially small terms until $t$ is exponentially large in $1/\ep$.

The following application of our main results quantifies this behavior in a simple way.
Let $\lambda^*_t \in \Delta_m$ denote the weights given by~\eqref{eq:lambda} for the measure $\mu_t$.
Define
\begin{equation}\label{eq:delta}
	\delta(t) := \sup_{s \leq t} \kl{\lambda^*_s}{w} - \kl{\lambda^*_t}{w}\,.
\end{equation}
If the dynamics exhibit metastability and the mixture components correspond to nearly invariant regions, then $\lambda^*_t$ varies only on long time scales and $\delta(t)$ is negligible until $t$ is exponentially large.
The following shows that slow transitions between regions automatically imply fast relaxation to metastable states.
\begin{theorem}\label{thm:meta}
	Let $\pi_1, \dots, \pi_m$ and $\pi$ satisfy \Cref{assume:main}.
	If $\pi_1, \dots, \pi_m$ satisfy \eqref{eq:mlsi} and~\eqref{eq:dirichlet_entropy}, then
	\begin{equation}
		\kl{\mu_t}{\cC} \leq e^{-t/C'} \kl{\mu_0}{\pi} + \delta(t)\,.
	\end{equation}
\end{theorem}
\begin{proof}
\Cref{thm:rlsi} implies
\begin{equation}
	\frac{d}{dt} \kl{\mu_t}{\pi} = - \cE\left(\radon{\mu_t}{\pi}, \log \radon{\mu_t}{\pi}\right) \leq - \frac 1{C'} \kl{\mu_t}{\pi} + \frac 1{C'} \kl{\lambda^*_t}{w}
\end{equation}
Applying Gronwall's inequality yields
\begin{align*}
	\kl{\mu_t}{\pi} & \leq e^{-t/C'} \kl{\mu_0}{\pi} + \frac 1{C'} e^{-t/C'} \int_0^t e^{s/C'} \kl{\lambda^*_s}{w} \dd s \\
	& \leq e^{-t/C'} \kl{\mu_0}{\pi} + \frac 1{C'} e^{-t/C'} (\kl{\lambda^*_t}{w} + \delta(t))\int_0^t e^{s/C'}  \dd s\\
	& \leq e^{-t/C'} \kl{\mu_0}{\pi} + \kl{\lambda^*_t}{w} + \delta(t)
\end{align*}
Hence
\begin{equation}
	\kl{\mu_t}{\cC} \leq 	\kl{\mu_t}{\pi} - \kl{\lambda^*_t}{w} \leq e^{-t/C'} \kl{\mu_0}{\pi} + \delta(t)\,,
\end{equation}
as claimed.
\end{proof}

This calculation also implies that the only obstruction to fast mixing is misallocation of the original weights.
\begin{corollary}
	Under the same conditions as \cref{thm:meta}, if $\sup_{s \leq t} \kl{\lambda^*_s}{w} \leq \eta$, then
\begin{equation}
	\kl{\mu_t}{\pi} \leq e^{-t/C'} \kl{\mu_0}{\pi} + \eta\,.
\end{equation}
\end{corollary}

\section{Duality theory and a dynamic characterization of transport-information inequalities}
Following~\citet{guillin-leonard-wu-yao}, we can use convex duality to give an attractive ``dual'' characterization of \cref{thm:main} in terms of concentration inequalities for the stochastic process associated with $\cL$.
We begin by developing a dual expression for the transport cost.
Recall the definition of the $c$-conjugate: for any $f: \cX \to \RR$, define its conjugate by
\begin{equation}
	f^c(y) := \inf_{x \in \cX} c(x, y) - f(x)\,.
\end{equation}
The following theorem provides a dual expression for $T_c(\mu, \cC)$.
\begin{theorem}\label{thm:dual}
	\begin{equation}
		T_c(\mu, \cC) = \sup_{f \in C_b(\cX)} \left\{\int f \dd \mu + \min_{i \in [m]} \int f^c \dd \pi_i\right\}
	\end{equation}
\end{theorem}
This expression shows that $T_c(\mu, \cC)$ has a dual formulation that looks exactly like the standard one in optimal transport, except that for each candidate dual solution the measure $\mu$ is compared not with $\pi$ but with the best option among $\{\pi_1, \dots, \pi_m\}$.
\begin{proof}
	By standard Kantorovich duality~\citep[Theorem 5.10]{Vil08}, we have
	\begin{equation}
		T_c(\mu, \cC) = \inf_{\tilde \pi \in \cC} \sup_{f \in C_b(\cX)} \int f \dd \mu + \int f^c \dd \tilde \pi\,.
	\end{equation}
	The set $\cC$ is a convex and compact subset of $\cP(\cX)$, and $(\tilde \pi, f) \mapsto \int f \dd \mu + \int f^c \dd \tilde \pi$ is convex-concave, so the minimax theorem applies~\citep{ky_fan} and we may interchange the order of $\inf$ and $\sup$ to obtain
	\begin{equation}
		T_c(\mu, \cC) = \sup_{f \in C_b(\cX)} \int f \dd \mu + \inf_{\tilde \pi \in \cC} \int f^c \dd \tilde \pi\,.
	\end{equation}
	Finally, linearity implies $\inf_{\tilde \pi \in \cC} \int f^c \dd \tilde \pi = \min_{i \in [m]} \int f^c \dd \pi_i$, 
proving the claim.
\end{proof}

\Cref{thm:dual} opens the door to using many of the tools developed by~\citet{guillin-leonard-wu-yao} for analyzing transport-information inequalities.
We state one such result in terms of a concentration inequality.
Write $\mathbb{P}_\mu(\cdot)$ for the distribution of the Markov process $(X_t)_{t \geq 0}$ corresponding to the generator $\cL$, initialized from $X_0 \sim \mu$.
\begin{theorem}\label{thm:equiv}
	Let $\pi_1, \dots, \pi_m$ and $\pi$ satisfy \Cref{assume:main}, and assume that the cost $c$ satisfies $c(x,x) = 0$ for all $x \in \cX$.
	The following are equivalent:
	\begin{enumerate}
		\item Inequality~\eqref{eq:rti} holds.
\item For any $r, t > 0$, $f \in C_b(\cX)$, and $\mu \in \cP(\cX)$ such that $\rd \mu/\rd \pi \in L^2(\pi)$,
		\begin{equation}\label{eq:conc}
			\mathbb{P}_\mu\left(\frac 1t \int_0^t f(X_s) \dd s  \geq r -  \min_{i \in [m]} \pi_i(f^c)\right) \leq \left\|\frac{\rd \mu}{\rd \pi} \right\|_{L^2(\pi)} \exp(-t \alpha(r))\,.
		\end{equation}
\end{enumerate}
\end{theorem}
This equivalence is particularly salient in the case where the cost in the definition of $T_c$ is a metric and $\alpha(x) = x^2/4C^2$, so that~\eqref{eq:rti} reads
\begin{equation}\label{eq:w1i}
	W_1^2(\mu, \cC) \leq 4 C^2 \FI{\mu}{\pi} \quad \forall \mu \in \cP(\cX)\,.
\end{equation}
\begin{corollary}
	Let $\pi_1, \dots, \pi_m$ and $\pi$ satisfy \Cref{assume:main}.
	The following are equivalent:
	\begin{enumerate}
		\item Inequality~\eqref{eq:w1i} holds.
\item For any $r, t > 0$, Lipschitz $f: \cX \to \RR$, and $\mu \in \cP(\cX)$ such that $\rd \mu/\rd \pi \in L^2(\pi)$,
		\begin{equation}\label{eq:dev}
			\mathbb{P}_\mu\left(\frac 1t \int_0^t f(X_s) \dd s  \geq \max_{i \in [m]} \pi_i(f) + r \right) \leq \left\|\frac{\rd \mu}{\rd \pi} \right\|_{L^2(\pi)} \exp\left(-\frac{t r^2}{4 C^2 \|f\|_{\mathrm{Lip}}^2}\right)\,.
		\end{equation}
\end{enumerate}
\end{corollary}
In other words, if the spatial average of the Lipschitz function is small under each of the components $\pi_i$, then the time average $\frac 1t \int_0^t f(X_s) \dd s$ is unlikely to be large.
Note that the results of \citet{lacker} can be used to show that a reweighted version of the $W_2$ transport-information inequality is equivalent to a dimension-free version of~\eqref{eq:dev} for i.i.d.\ copies of the Markov process.
\begin{proof}
	We may assume by rescaling that $\|f\|_{\mathrm{Lip}} = 1$, in which case the result follows directly from \cref{thm:equiv} and the fact that $f^c = -f$ when $f$ is $1$-Lipschitz and $c$ is a metric, so $- \min_{i \in [m]} \pi_i(f^c) = \max_{i \in [m]} \pi_i(f)$.
\end{proof}
We now prove \cref{thm:equiv}.
\begin{proof}[Proof of \cref{thm:equiv}]
	We proceed exactly as in \citet{guillin-leonard-wu-yao}.
	Assume that \eqref{eq:rti} holds.
	Given $f \in C_b$, let $I_f(s) := \inf \{\FI{\mu}{\pi} : \mu(f) = s\}$.
	Then for any $\mu$ for which $\mu(f) + \min_{i \in [m]} \pi_i(f^c) \geq 0$,
\eqref{eq:rti} and \cref{thm:dual} imply
	\begin{equation*}
		\FI{\mu}{\pi} \geq \alpha(T_c(\mu, \cC)) \geq \alpha(\mu(f) + \min_{i \in [m]} \pi_i(f^c))\,.
	\end{equation*}
	We obtain that $I_f(r - \min_{i \in [m]} \pi_i(f^c)) \geq \alpha(r)$ for all $r > 0$.
	Wu's non-asymptotic large deviations bound~\citep[Theorem 1]{wu-deviation} shows that for any $s > 0$,
	\begin{equation}\label{eq:ldp}
			\mathbb{P}_\mu\left(\frac 1t \int_0^t f(X_s) \dd s \geq \pi(f) + s\right) \leq \left\|\frac{\rd \mu}{\rd \pi}\right\| _{L^2(\pi)}\exp(-t \lim_{\delta \to 0} I_f(\pi(f) + s - \delta))
	\end{equation}
	The assumption that $c(x,x) = 0$ for all $x \in \cX$ implies that $f \leq - f^c$; hence $\pi(f) = \sum_{i=1}^m w_i \pi_i(f) \leq - \min_{i \in [m]} \pi_i(f^c)$.
	Therefore we can choose $s = r - \pi(f) -  \min_{i \in [m]} \pi_i(f^c) > 0$ to obtain
	\begin{equation*}
		\mathbb{P}_\mu\left(\frac 1t \int_0^t f(X_s) \dd s \geq r -  \min_{i \in [m]} \pi_i(f^c)\right) \leq \left\|\frac{\rd \mu}{\rd \pi}\right\| _{L^2(\pi)}\exp(-t \lim_{\delta \to 0} \alpha(r - \delta))\,.
	\end{equation*}
	The lower-semicontinuity of $\alpha$ implies the claim.
	
	Conversely, assume that~\eqref{eq:conc} holds.
	If $T_c(\mu, \cC) = 0$, then~\eqref{eq:rti} is vacuous, so assume that $T_c(\mu, \cC)$ is strictly positive.
	Since the Markov process is symmetric, \eqref{eq:ldp} is asymptotically tight and we have the lower bound~\citep[Theorem B.1 and Corollary B.11]{wu-ui}
	\begin{equation*}
		\liminf_{t \to \infty} \frac 1t\log \mathbb{P}_\pi\left(\frac 1t \int_0^t f(X_s) \dd s \geq \pi(f) + s\right) \geq - \inf\{\FI{\tilde \mu}{\pi} : \tilde \mu(f) > \pi(f) + s\} \quad \forall s \in \RR
	\end{equation*}
	This fact combined with~\eqref{eq:conc} implies for any $r > 0$ and any $f \in C_b$ that
	\begin{equation*}
		\inf\{\FI{\tilde \mu}{\pi} : \tilde \mu(f) +\min_{i \in [m]} \pi_i(f^c) > r\} \geq \alpha(r)\,.
	\end{equation*}
	By \cref{thm:dual}, we can choose a sequence $(f_n)_{n \geq 1} \subseteq C_b$ such that $\mu(f_n) +\min_{i \in [m]} \pi_i(f_n^c) > T_c(\mu, \cC) - 1/n$, so we obtain $\alpha(T_c(\mu, \cC) - 1/n) \leq \FI{\mu}{\pi}$.

	Letting $n \to \infty$ and applying lower-semicontinuity of $\alpha$ yields the claim.
\end{proof}

\bibliographystyle{abbrvnat}

\subsection*{Acknowledgments}
The author is grateful to Sinho Chewi and Shay Sadovsky for conversations related to this project, as well as to the technical assistance of ChatGPT-5.2-Pro and Claude Opus 4.6. This work was supported by National Science Foundation grant DMS-2339829.
\end{document}